\documentclass[prx,aps,superscriptaddress,twocolumn,nofootinbib,longbibliography]{revtex4-2}

\usepackage{graphicx,color,amsmath,amsfonts,enumerate,amsthm,amssymb,mathtools,enumitem,thmtools,hyperref,subfigure,mathdots,enumitem,centernot,bm,soul,bbm,tikz,pgfplots,float}
\usepackage[capitalise, noabbrev]{cleveref}
\usepackage[normalem]{ulem}
\usepackage{tcolorbox}
\hypersetup{colorlinks=true,linkcolor=teal,citecolor=teal,urlcolor=teal}
\pgfplotsset{compat=newest}
\tcbset{before skip=10pt,toptitle=2mm,bottomtitle=1mm,fonttitle=\bfseries}
\tcbuselibrary{theorems}
\tcbuselibrary{breakable}

\definecolor{NavyBlue}{rgb}{0.0, 0.0, 0.5}
\definecolor{OliveGreen}{rgb}{0.33, 0.42, 0.18}
\definecolor{def_color_frame}{RGB}{220,230,242}
\colorlet{def_color_back}{def_color_frame!30}
\definecolor{def_color_text}{RGB}{37,64,97}
\definecolor{def_color_frame2}{RGB}{242,200,200}
\colorlet{def_color_back2}{def_color_frame2!30}
\definecolor{def_color_text2}{RGB}{97,55,33}

\newtcolorbox[auto counter]{bluebox}[2][]{%
colback=def_color_back,colframe=def_color_frame,fonttitle=\bfseries,coltitle=def_color_text,float,floatplacement=t,title=Box~\thetcbcounter: #2,#1}

\newtcolorbox[use counter from=bluebox]{redbox}[2][]{%
colback=def_color_back2,colframe=def_color_frame2,fonttitle=\bfseries,coltitle=def_color_text2,float,floatplacement=t,title=Box~\thetcbcounter: #2,#1}

\def\E{ {\cal E} }

\def\T{ {\cal T} }

\def \Tr{\text{Tr}}

\def\>{\rangle}
\def\<{\langle}
\newcommand{\bra}[1]{\langle {#1} |}
\newcommand{\ket}[1]{| {#1} \rangle}
 
\newcommand{\ketbra}[2]{\ensuremath{|#1\rangle\!\langle#2|}}

\newcommand{\iden}{\mathbbm{1}}

\renewcommand{\v}[1]{\ensuremath{\boldsymbol #1}}

\definecolor{bluecyan}{rgb}{0.27, 0.66, 0.88}
\definecolor{ppblue}{RGB}{46,117,182}
\definecolor{ppred}{RGB}{197, 90, 17}

\theoremstyle{plain}
\newtheorem{thm}{Theorem}

\newtheorem{lem}[thm]{Lemma}

\theoremstyle{definition}

\newcommand{\tanmoy}[1]{ \textcolor{violet}{{\tt TB}: #1}}
\newcommand{\chandan}[1]{ \textcolor{brown}{{\tt CD}: #1}}

\usepackage{newfloat}
\usepackage[super]{nth}
\DeclareFloatingEnvironment[fileext=frm,placement={!ht},name=Box]{myfloat}

\definecolor{beige}{rgb}{1.00,0.95,0.90}
\usepackage[framemethod=TikZ]{mdframed} 
\newenvironment{frameenv}[1]
    {\begin{myfloat}[]
    \begin{mdframed}[roundcorner=4pt,backgroundcolor=blue!5]
    \caption{\bfseries #1}
    }
    { \vspace{0.3cm}
    \end{mdframed}\end{myfloat}
    }

\begin{document}

\title{Optimal performance of a three stroke heat engine in the microscopic regime}

\author{Tanmoy Biswas}
\affiliation{Theoretical Division (T4), Los Alamos National Laboratory, Los Alamos, New Mexico 87545, USA.}
\affiliation{International Centre for Theory of Quantum Technologies, University of Gdansk, Wita Stwosza 63, 80-308 Gdansk, Poland.}
\author{Chandan Datta}
\affiliation{Institute for Theoretical Physics III, Heinrich Heine University D\"{u}sseldorf, Universit\"{a}tsstra{\ss}e 1, D-40225 D\"{u}sseldorf, Germany}
\affiliation{Department of Physics, Indian Institute of Technology Jodhpur, Jodhpur 342030, India}

\begin{abstract}

We consider a three-stroke engine in the microscopic regime, where the working body of the engine is composed of a two-level system. The working body of the engine aims to withdraw heat from the hot heat bath, generate work, and discharge the surplus heat into the cold heat bath through the successive execution of three strokes. In this process, the interaction of the working body with the heat baths is assumed to be energy-conserving and thus can be described by thermal operations. While earlier studies analyzed the optimal performance of this engine when the working body could be transformed by any arbitrary thermal operation, we present closed expressions for the maximum work produced by the engine and the maximum efficiency of the engine when only a restricted class of thermal operations can be implemented on the working body. Furthermore, we explore the engine's optimal performance under two well-studied classes of restrictions: thermal operations realized via Jaynes-Cummings interaction and thermal operations realizable with finite-sized heat baths. Therefore, on one hand, our results are general, as they reproduce the optimal performance achieved when any arbitrary thermal operation can be implemented on the working body once the restriction is relaxed. On the other hand, our results allow us to determine the engine's maximum work production and efficiency in a more realistic scenario, where only a restricted class of thermal operations are possible, thereby bringing our findings closer to experimental feasibility. 
\end{abstract}
\maketitle

\section{Introduction}

Within the domain of thermodynamics, heat engines play a pivotal role by converting thermal energy into work. The study of heat engines not only propels technological advancements in energy conversion but also deepens our understanding of nature's fundamental physical laws. However, the advent of quantum mechanics has ushered in a new era of insight, particularly concerning the foundational underpinnings of thermodynamics on a microscopic scale. This captivating avenue of inquiry traces its origin back to the 1950s, when investigations into the thermodynamic properties of lasers commenced \cite{Ramsey22, Scovil1959, Scovil2}. Since then, a resounding interest has flourished in comprehending the mechanisms of quantum thermal engines as well as quantum thermal machines in general \cite{Kosloff1984, Kosloff2014, Scully2002, PopescuRefrigerator, Allahverdyan2004, Segal, Mahler, PopescuSmallestHeatEngine2010, BrunnerVirtualqubit2012, Alicki2004, Zhang2022, Myers2022, Mukherjee_2021,Bhattacharjee2021, BiswasFDR2022, BiswasLobejko, Biswas2022extractionof, BiswasNew, SahaBhattacharya}, and in formulating microscopic thermodynamic frameworks to characterize their operations \cite{Pusz1978, Alicki_1979, Davies_1978, Esposito2009, horodecki2013fundamental, Skrzypczyk2014}. Recently, researchers have realized microscopic heat engines in the lab using a trapped single calcium ion as a working body \cite{Blickle2011, Abah_2012, Rossnagel2016, Lindenfels2019, Maslennikov_2019}, along with superconducting circuits \cite{NYHSG}, nitrogen vacancy centers in diamonds \cite{Klaztow2019}, atomic \cite{Bouton_2021} and phononic systems \cite{Zanin_2022}, and single-electron transistors \cite{Koski_2014, Koski_2015}. 

The microscopic heat engines can be described as thermodynamic devices that involve small-sized components. In particular, the microscopic heat engines described in the literature involve working bodies that are finite-dimensional and can be as simple as two-level systems \cite{Alicki_1979,Myers2022}. Furthermore, the analysis of the microscopic heat engines has been extended further by considering heat baths of finite size \cite{Lobejko2021, Tajima_Hayashi_PRE, Ito_Hayashi_PRE, Pozas-Kerstjens_2018, Mohammady_PRE}. Based on their way of functioning, the heat engines considered in the microscopic regime can be broadly classified into two classes: continuous engines \cite{PopescuSmallestHeatEngine2010, PopescuRefrigerator, MitchisonContemp, Nimmrichter2017, Kosloff2014, Bera2021} and stroke-based engines \cite{Niedenzu2019, Ptas22, Zhang2014_opto, UzdinLevyKosloff2015}. In a continuous engine, the working body is always in contact with the hot and cold heat baths. As a consequence, the engine withdraws heat from the hot heat bath and converts it into work in a continuous and simultaneous manner. On the other hand, in a stroke-based heat engine, the working body exchanges heat and converts it into work in discrete strokes that are well separated in time. Therefore, in stroke-based engines, the working body repeatedly couples between the hot and  cold baths using subsequent \emph{strokes}. 


In this paper, we examine a three stroke heat engine that operates in the microscopic regime introduced in Ref. \cite{Lobejko2020}.  By microscopic, we imply that the working body of this engine is a two-level system. The thermodynamic cycle of the engine is composed of three strokes. In the first stroke, the working body aims to withdraw heat from the hot heat bath, followed by the second and third strokes, where the working body converts the heat into work, disposes of the extra amount of heat in the cold heat bath, and completes the cycle. We assume that the interaction between the working body and heat baths is energy-conserving to prevent any unaccounted energy sources from contributing to the work done by the engine. In other words, the engine produces work by harnessing the temperature difference between the two heat baths. This led us to describe the transformation of the state of the working body via a \emph{thermal operation} \cite{Janzing2000, horodecki2013fundamental, Lostaglio_2019}. The extraction of work is depicted by a unitary process in which the engine's working body converts a fraction of the heat extracted from the hot heat bath into work, while the remaining heat is transferred to the cold heat bath.

The three-stroke engine discussed in Ref. \cite{Lobejko2020} presents  optimal expressions of the work produced by the engine and its efficiency, assuming that the working body can undergo any arbitrary thermal operation during the first and third strokes. The experimental implementation of any arbitrary thermal operation poses a formidable challenge since the realization of any arbitrary thermal operation assumes infinitely large sized baths with exponentially large amount of degeneracy in the spectrum. Therefore, in a practical scenario, only a restricted class of thermal operations can be realized. This restriction depends on the experimental setup used to realize the three stroke engine. Therefore, achieving the maximum work produced by the engine and the maximum efficiency, as demonstrated in Ref. \cite{Lobejko2020}, is very challenging in a realistic scenario. 

 In this paper, we provide expressions for the maximum work produced by the three-stroke engine and maximum efficiency of it when the engine's working body undergoes a restricted class of thermal operations. On one hand, our results reproduce the maximum amount of extracted work and the efficiency obtained in Ref. \cite{Lobejko2020} when the restrictions are relaxed. On the other hand, our results allow us to determine the maximum work and efficiency when the thermal operations acting on the working body are restricted to those that can be realized via the Jaynes-Cummings interaction between a two-level working body and a heat bath modeled with a resonant single-mode harmonic oscillator. In addition, we extend our study to a scenario where the restriction on the thermal operations acting on the two-level working body arises due to the presence of a finite-dimensional heat bath modeled with a truncated single-mode harmonic oscillator, rather than an infinitely large heat bath, which is typically required to implement any arbitrary thermal operation.  

The paper is organized as follows. In Sec. \ref{Description}, we describe the construction of the stroke-based heat engine. We outline the individual strokes that constitute the thermodynamic cycle for this specific three-stroke heat engine. The mathematical framework that describes the characteristics of the strokes is also presented. Subsequently, we establish the thermodynamic framework to precisely describe the process of extracting work and exchanging heat in this stroke-based heat engine. In Sec. \ref{Main_Results}, we calculate the maximum work production and efficiency. The results are obtained by considering a general range of thermal operations that describe the transformation of the working body when it comes into contact with the heat bath. Using the general expressions of optimal work production and efficiency, we examine the work production and efficiency for specific instances, such as when the working body's interaction with heat baths is modelled by the Jaynes-Cummings interaction and the heat bath's dimension is finite (modelled by a truncated harmonic oscillator). Finally, we conclude in Sec. \ref{conclusion}.

\section{Description of a stroke-based discrete engine}\label{Description}
 
\subsection{Construction of the engine } 
Any stroke-based heat engine operates in a sequence of strokes that allows it to consume heat from a hot heat bath and convert it into work in subsequent strokes. The second law of thermodynamics asserts that it is impossible to convert all the heat that has been drawn from the hot heat bath to work. Thus, in order to work in a cyclic manner, the engine needs to release a fraction of the consumed heat into the cold heat bath after work extraction in the final stroke. Here, we consider a model of a stroke-based engine in the microscopic regime that requires three subsequent strokes to complete the thermodynamic cycle. The heat engine consists of the following key components:
\begin{enumerate}
    \item A hot heat bath is thermalized at an inverse temperature $\beta_H$, which serves as a source of heat that requires to produce work. We label this component by $H$.
    \item A cold heat bath is thermalized at inverse temperature $\beta_C$ and works as a heat sink for dumping the remainder of the consumed heat from the hot heat bath after work has been extracted. We denote this component by $C$.
    \item A working substance or working body represented by a $d$-level quantum system that controls the flow of energy among other components of the engine. We label this component by $S$.
\end{enumerate}
Therefore, the total free Hamiltonian of the whole system is
\begin{equation}
    H_{\text{free}} = H_S+H_H+H_C,
\end{equation}
where $H_S$, $H_H$, and $H_C$ denote the Hamiltonian of the working body, hot heat bath, and cold heat bath, respectively. In the start, we assume all the components of the engine are uncorrelated, which allows us to write the full initial state $\rho_{\text{in}}$ of the engine as 
\begin{equation}
    \rho_{\text{in}} = \rho_S\otimes \gamma_{H}\otimes \gamma_{C},
\end{equation}
 where $\gamma_H$ and $\gamma_C$ denote the thermal states of the hot and cold heat baths, i.e.,
\begin{equation}\label{Gibbs_defn}
    \gamma_H = \frac{e^{-\beta_H H_H}}{\Tr(e^{-\beta_H H_H})};\quad \gamma_C = \frac{e^{-\beta_C H_C}}{\Tr(e^{-\beta_C H_C})}.
\end{equation}
Whereas, $\rho_S$ represent the initial state of the working body. In the following section, we shall discuss the functioning of the engine.

\subsection{Characterization of the strokes of the engine}
\label{Characterization}
This section focuses on the various strokes that allow the engine to operate. The engine operates on three major strokes: i) heat stroke, ii) work stroke, and iii) cold stroke. While the engine involves energy-conserving stroke operations, this does not result in a fully autonomous engine since it needs external control to carry out the various steps. Nevertheless, the engine is an advancement towards creating an autonomous machine as it preserves total energy. To put it simply, the process of enabling or disabling interactions within the system does not impact the flow of energy into or out of the system. Following this, we describe the workings of the engine in more detail. 

\subsubsection{Heat stroke}

This is the initial stroke of the engine, in which the working body interacts with the hot heat bath at thermal equilibrium using energy-conserving unitary $U_{S,H}$, i.e, $[U_{S,H},H_{S}+H_{H}] = 0$. Thus, the transformation of the state $\rho_S$ is governed by a thermal operation $\E_{H}$ \cite{horodecki2013fundamental}, i.e.,
\begin{equation}
    \rho_S\rightarrow \E_{H}(\rho_S) := \Tr_{H}\big(U_{S,H}(\rho_S\otimes \gamma_{H})U_{S,H}^{\dagger}\big). 
\end{equation}
The aim of this stroke is to consume heat from the hot heat bath, which increases average energy and induces non-passivity in the working body. In accordance with classical thermodynamics, where heat is defined as the amount of energy transferred to the system from the hot heat bath, we define the amount of heat transferred, denoted as $Q_H$, as the change in the system's average energy caused by its interaction with the hot heat bath, which can be expressed as:
\begin{equation}\label{Defn:HeatExchanged}
Q_H= \Tr\Big(H_S\big(\mathcal{E}_{H}(\rho_S)-\rho_S\big)\Big).
\end{equation}
Here, $\mathcal{E}_{H}(\rho_S)$ represents the final state of the working body after the interaction with the hot heat bath.

\subsubsection{Work stroke}\label{work_stroke}

In the second stroke of the engine, the working body, which resulted from the first stroke, aims to produce work using the heat it drew from the first stroke.  The engine extracts work using a \emph{cyclic}  unitary operation that acts on the working body's state from the first stroke, i.e., $\E_{H}(\rho_S)$, and transforms it to the one with lower average energy \cite{Allahverdyan2004}. A unitary operation is called cyclic if it is generated by a time-dependent Hamiltonian $G(t)$ such that $G(0)=G(\tau)=H_S$, where $\tau$ denotes the duration of the work stroke and $H_S$ is the Hamiltonian of the working body. We denote the set of such cyclic unitary operations by $\mathcal{V}_{\text{cyc}}$. \\
In the following, we emphasize an important distinction between our engine and the familiar four-stroke engine operating in Otto cycle (See Fig. 1 (b) of Ref. \cite{cangemni_Levy_engines}). In a four-stroke engine operating in an Otto cycle, the thermodynamic cycle consists of two work strokes: the expansion stroke and the compression stroke. After the first work stroke, the free Hamiltonian of the two-level working body changes, and after the second work stroke, it returns to its initial value. In contrast, the proposed three-stroke engine operates with a thermodynamic cycle that includes only one work stroke that extracts work by switching on an interaction that alters the external Hamiltonian of the engine's working body. The interaction is switched off upon completion of the work stroke. Thus, the Hamiltonian of the working body is reduced to the free Hamiltonian at the end of the work stroke, which is the same as the Hamiltonian at the beginning of the work stroke. In other words, the external driving parameter undergoes cyclic changes, meaning it reverts back to its original value at the end of the work stroke \cite{Allahverdyan2004, Alicki_1979}.

Therefore, this stroke transforms the working body's state as follows: 
\begin{equation}
    \mathcal{E}_{H}(\rho_S)\rightarrow V (\mathcal{E}_{H}(\rho_S)) V^{\dagger},
\end{equation}
where $V$ is a cyclic unitary operation. The amount of work produced by the engine is defined as \cite{Biswas2022extractionof, Lobejko2020}
\begin{equation}\label{Exp_Work}
    W = \Tr\Big(H_S\big(\mathcal{E}_{H}(\rho_S)- V (\mathcal{E}_{H}(\rho_S)) V^{\dagger}\big)\Big).
\end{equation}
 To maximize the work $W$ over all cyclic unitary operations, we describe the concept of \emph{ergotropy} introduced in the seminal Ref. \cite{Allahverdyan2004} by Allahverdyan et al. Consider any arbitrary $d$-dimensional state $\sigma$ 
having eigendecomposition
\begin{equation}
    \sigma = \sum_{i=0}^{d-1}\lambda_i\ketbra{\lambda_i}{\lambda_i}\quad \text{such that}\quad \forall{i},\;\; \lambda_{i}\geq \lambda_{i+1},
\end{equation}
associated with Hamiltonian $\mathcal{H}$ with eigendecomposition
\begin{equation}
    \mathcal{H} = \sum_{j=0}^{d-1}\varepsilon_j\ketbra{\varepsilon_j}{\varepsilon_j} \quad \text{such that}\quad \forall{j},\;\; \varepsilon_{j}\leq \varepsilon_{j+1}.
\end{equation}
Then the ergotropy of the state $\sigma$ defined as follows:
\begin{eqnarray}
    R(\sigma)&:=&\max_{\mathcal{U}\in \mathcal{V}_{\text{cyc}}}\Tr\left(\mathcal{H}\left(\sigma-\mathcal{U}\sigma \mathcal{U}^{\dagger}\right)\right) \nonumber\\&=& \Tr\left(\mathcal{H}\left(\sigma-\mathcal{U}_{\text{max}}\sigma \mathcal{U}_{\text{max}}^{\dagger}\right)\right),
\end{eqnarray}
where $\mathcal{U}_{\text{max}}$ has the following form:
\begin{equation}
    \mathcal{U}_{\text{max}} = \sum_{j=0}^{d-1}\ketbra{\varepsilon_{j}}{\lambda_{j}}.
\end{equation}
By inspection, we observe that $\mathcal{U}_{\text{max}}$ decreasingly orders the eigenvalues of $\sigma$ such that they correspond to the eigenlevels of Hamiltonian $\mathcal{H}$ arranged in increasing order of energy.

Employing the concept of ergotropy, we can now conclude that when the state of the working body in state $\E_{H}(\rho_S)$ transforms via a cyclic unitary operation $V$, the work production by the engine satisfies the following condition:
\begin{eqnarray}
    W &=& \Tr\Big(H_S\big(\mathcal{E}_{H}(\rho_S)- V (\mathcal{E}_{H}(\rho_S)) V^{\dagger}\big)\Big) \\&\leq& \max_{V\in \mathcal{V}_{\text{cyc}}}\Tr\Big(H_S\big(\mathcal{E}_{H}(\rho_S)- V (\mathcal{E}_{H}(\rho_S)) V^{\dagger}\big)\Big)
    \\&=&  R(\mathcal{E}_{H}(\rho_S)).
\end{eqnarray}
In other words, the maximum amount of the work produced by the engine is simply equals to the ergotropy of the state of the working body after the first stroke.

\subsubsection{Closing the thermodynamic cycle}

In the final stroke of the engine, the state of the working body that results from the work stroke $V\E_H(\rho_S)V^{\dagger}$ interacts with the cold heat bath in thermal equilibrium via an energy conserving unitary operation $U_{S,C}$, i.e., $[U_{S,C},H_S+H_C]=0$. This transformation corresponds to a thermal operation with respect to the cold heat bath denoted by $\E_{C}$, i.e.,
\begin{equation}\label{third_stroke_eqn}
     \mathcal{E}_{C}\big(V\E_H(\rho_S)V^{\dagger}\big) :=  \Tr_{C}\Big(U_{S,C}\big((V\E_{H}(\rho_S)V^{\dagger})\otimes \gamma_{C}\big)U_{S,C}^{\dagger}\Big).
\end{equation}

The goal of this stroke is to return the working body to its initial state, i.e., 
\begin{equation}\label{cyclic_state}
   \mathcal{E}_{C}\big(V\E_H(\rho_S)V^{\dagger}\big)=\rho_S,
\end{equation}
that results in the closing of the thermodynamic cycle. Our aim is to figure out the \emph{cyclic state} $\rho_S$ that satisfies the relation given in Eq. \eqref{cyclic_state}.

We define \emph{cold heat} as the amount of heat that the working body releases into the cold heat bath, expressed as follows:

\begin{equation}\label{cold_heat}
    Q_C = \Tr\Big( H_S\big(\rho_S-V\E_H(\rho_S)V^{\dagger}\big)\Big).
\end{equation}

This step is critical for the entire setup to function as an engine. More generally, this stroke encodes the manifestation of the second law for this heat engine. 

\subsection{The open cycle engine}\label{Open_cycle_engine}

The open cycle engine introduced in Ref. \cite{Biswas2022extractionof} represents a subclass of the discussed three-stroke engine. This kind of engine starts with a working body that is thermalized at the inverse temperature of the cold heat bath. Therefore, the final stroke of the engine simply reduces to thermalization at the inverse temperature of the cold bath, which results in the completion of the cycle. Our research aims to bridge the understanding between an open-cycle engine and a generic three-stroke engine that we have introduced previously, particularly when the working body comprises two-level systems. This connection significantly simplifies the optimization of efficiency and work for the three-stroke engines with a two-level working body. The comparison between these two engines has been illustrated in Fig. \ref{fig:Open_vs_closed}. 

\begin{figure*}[t]
    \centering
    \includegraphics[width=12 cm]{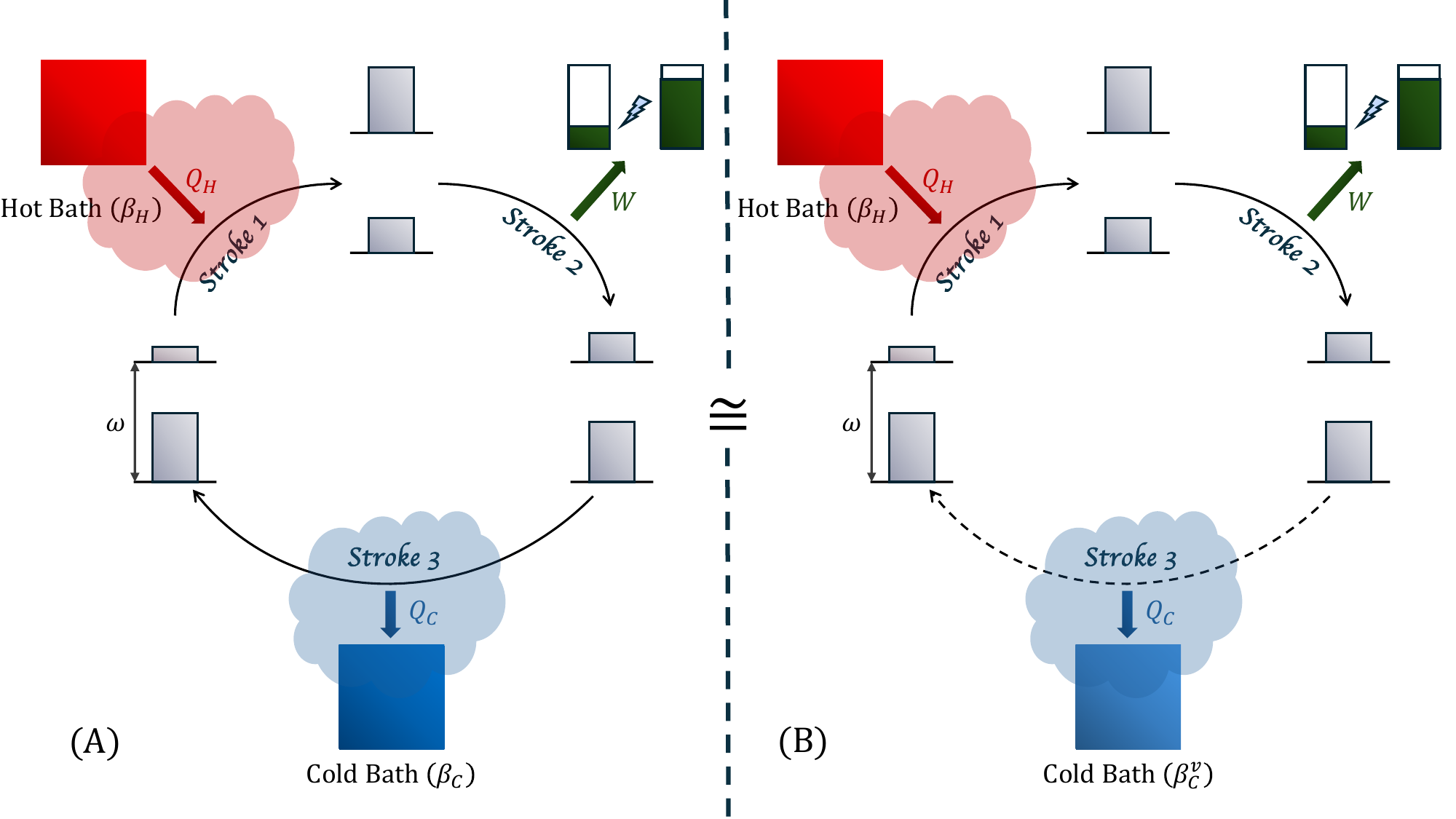}
    \caption{  The above figure depicts the equivalence between the proposed three-stroke engine and the open-cycle engine introduced in Sec. \ref{Open_cycle_engine}. On the left panel, we show the proposed three-stroke engine, where the first and third strokes aim to withdraw heat from the hot bath and dispose of the excess heat in the cold bath, whereas the production of work happens during the second stroke. The extracted amount of work can be manifested in terms of the charging of the battery (in this work, we have not considered any explicit battery). The first and third strokes are characterised by thermal operations at the hot and cold inverse temperatures $\beta_H$ and $\beta_C$, respectively. This proposed engine is equivalent to an open-cycle engine, where the third stroke is characterized by thermalization at some effective inverse temperature $\beta^v_{C}$, as shown in the right panel. This equivalence captures the underlying idea behind the proof of theorem \ref{thm_optimal_work_eff}. }
    \label{fig:Open_vs_closed}
\end{figure*}

\subsection{Mathematical tools}\label{Mathematical_tools}

This section provides a brief overview to the mathematical tools associated with thermal operations that are important to evaluate the performance of the engine. A thermal operation can be defined as a complete positive trace preserving (CPTP) map  $\E$, which transforms a system's state $\rho$ as 
\begin{equation}
    \mathcal{E}(\rho) = \Tr_{E}\Big(U(\rho\otimes \gamma_E)U^{\dagger}\Big),
\end{equation}
where $U$ satisfies the relation $[U,H_S+H_E]=0$ with $H_S$, $H_E$ are the Hamiltonians of the system and bath, respectively, and $\gamma_E=\frac{e^{-\beta H_E}}{\Tr(e^{-\beta H_E})}$ is the thermal state of the bath that is thermalized at an inverse temperature $\beta$ \cite{horodecki2013fundamental,Lostaglio_2019}.

The central task is then described by the following interconversion problem: what is the set of final states that can be achieved from a given state $\rho$ via thermal operations? This set is convex and referred to as the \emph{thermal cone} of the state $\rho$, denoted by $\T(\rho)$, which follows straightforwardly from the convexity of thermal operation \cite{LostaglioKorzekwaPRX}. Characterizing the necessary and sufficient conditions for the interconversion problem in general is challenging and unknown beyond dimension two \cite{Studzinski,Lostaglio2019}. Nonetheless, the seminal work by Horodecki and Oppenheim in Ref. \cite{horodecki2013fundamental} characterizes the necessary and sufficient conditions for the state interconversion problem when $[\rho,H_S] = 0$. These states are referred to as \emph{energy-incoherent states} as they do not possess coherence in the energy eigenbasis. For a given energy incoherent state $\rho$, the thermal cone $\T(\rho)$ is reduced to a polytope, i.e., the set will have a finite number of extremal points,  as coherence in the energy eigenbasis can not be generated under thermal operations (See theorem 14 of Ref. \cite{Lostaglio_2019}, \cite{LostaglioKorzekwaPRX}). In this case, we refer the set $\T(\rho)$ as the \emph{thermal polytope}. The thermal operation that transforms the initial energy incoherent state $\rho$ to an extremal point of the thermal polytope $\T(\rho)$ is called an \emph{extremal thermal operation}. Later, we shall see that these extremal points of the set $\T(\rho)$ and the extremal thermal operations play a crucial role in optimizing the performance of the three stroke engine. 

Let us  briefly describe the necessary and sufficient conditions for the interconversion problem for energy-incoherent states. In order to do so, we associate a probability vector $\v p$, called the \emph{population vector}, with an energy incoherent state $\rho$ such that $p_i=\langle E_i|\rho|E_i\rangle$, where $|E_i\rangle$ is an eigenstate of Hamiltonian $H_S$ with eigenvalue $E_i$. Then, an energy incoherent state $\rho$ can be converted to another energy incoherent state $\sigma$ via a thermal operation $\E$ iff there exists a stochastic matrix $T^{\E}$ such that 
\begin{equation}
    T^{\E}\v p = \v q\quad\quad T^{\E}\v{\gamma} = \v{\gamma},
\end{equation}
where $\v p$, $\v q$ and $\v{\gamma}$ are the population vectors associated with $\rho$, $\sigma$, and $\gamma$, respectively. Here, $\gamma$ is the Gibbs state of the system. We call such a stochastic matrix $T^{\E}$ a \emph{Gibbs-stochastic matrix} or \emph{thermal process} \cite{Lostaglio_2019,Mazurek_2018}. Employing the concept of relative majorization, one can prove that the existence of such a Gibbs-stochastic matrix is equivalent to the thermomajorization condition, which we have described in Appendix \ref{beta_order_section} \cite{Lostaglio_2019, Entropy_sagawa, Marshall1980InequalitiesTO}. The set of all thermal processes for a given dimension is a polytope. Among the extremal thermal processes, we shall be interested in those thermal processes that correspond to the extremal thermal operation. 

Let us illustrate the above-mentioned concepts for a two-level energy incoherent state, which will be crucial for optimizing the performance of the heat engine. For an energy incoherent initial state 
\begin{equation}\label{Eq:defrho}
    \rho = \begin{pmatrix}
        p & 0\\
        0 & 1-p
    \end{pmatrix},
\end{equation}
with Hamiltonian 
\begin{equation}
    H_S = \omega\ketbra{1}{1},
\end{equation}
we associate the population vector 
\begin{equation}
    \v p = \begin{pmatrix}
        p & 1-p
    \end{pmatrix}^T.
\end{equation}
The extremal thermal process in dimension $2$ is given by \cite{horodecki2013fundamental}: 
\begin{equation}
    \quad\Bigg\{\begin{pmatrix}
        1 & 0 \\
        0 & 1
    \end{pmatrix},\begin{pmatrix}
        1-e^{-\beta\omega} & 1\\
        e^{-\beta\omega} & 0 \\
    \end{pmatrix} \Bigg\}.
\end{equation}
Therefore, the thermal polytope for the state $\rho$ in Eq. \eqref{Eq:defrho} is given by 
\begin{equation}\label{convex_hull}
    \T (\rho) = \text{Conv}\Bigg\{\rho, \begin{pmatrix}
        1-pe^{-\beta\omega} & 1\\
        pe^{-\beta\omega} & 0
    \end{pmatrix}\Bigg\},
\end{equation}
where $\text{Conv} \{A\}$ denotes the convex hull of the set $A$. Thus, any arbitrary element of the set $\T(\rho)$ can be expressed as:
\begin{equation}\label{Range_of_lambda}
    \lambda \begin{pmatrix}
        1-pe^{-\beta\omega} & 1\\
        pe^{-\beta\omega} & 0
    \end{pmatrix}+ (1-\lambda) \rho 
\end{equation}
with $\lambda\in[0,1]$.

Here, we would like to make an important remark. When we define thermal operations, we suppose that our system of interest is in contact with a very large bath that has an effectively infinite volume. In such an ideal scenario, all possible states within the set $\T(\rho)$ can be realized. In a more practical context, various restricted subsets of states emerge within the thermal polytope $\T(\rho)$ as a result of constraints involving the feasibility of thermal operations. In such cases, the range of $\lambda$ in Eq. \eqref{Range_of_lambda} is more generally, confined between 
\begin{eqnarray}\label{Range_of_lambda_gen}
    [0,\lambda^{\text{max}}_{\beta}],  \;\text{where}\; 0\leq\lambda^{\text{max}}_{\beta}\leq 1.
\end{eqnarray}
In this paper, two particular subsets that will be of interest to us are as follows:
\begin{enumerate}
    \item The set of states that can be achieved from a given state of two-level systems through thermal operations in the presence of a finite heat bath, which can be realized with a truncated harmonic oscillator. \cite{MullerPastena2016}.
    \item The set of states which are achievable from a given state of two-level systems via thermal operations that can be realized via rotating wave approximated Jaynes-Cummings interaction \cite{Lostaglio2018elementarythermal}.  
\end{enumerate}

After explaining all of the necessary concepts related to the functioning of the three-stroke engine, we will now proceed by describing the main results of this work.

\section{Main Results}\label{Main_Results}
\subsection{Thermodynamic framework}
In this section, we shall develop the thermodynamic framework of the proposed heat engine. From the description of the three stroke heat engine that we have discussed in the previous section (see Ref. \ref{Description}), we can summarize the thermodynamic cycles of the three stroke engine as 
\begin{equation}\label{Flowchart}
    \rho_S\rightarrow\E_{H}(\rho_S)\rightarrow V\E_{H}(\rho_S)V^{\dagger}\rightarrow\E_{C}\big(V\E_{H}(\rho_S)V^{\dagger}\big)=\rho_S.
\end{equation}
From the flowchart of Eq. \eqref{Flowchart}, we can write the total change in average energy as 
\begin{eqnarray}
    0 &=& E(\rho_S)-E(\rho_S) \nonumber\\
    &=& E(\rho_S)-E(\E_{H}(\rho_S))+E(\E_{H}(\rho_S))-E(V\E_{H}(\rho_S)V^{\dagger})\nonumber\\&&+ E(V\E_{H}(\rho_S)V^{\dagger})-E(\rho_S) \nonumber\\
    &=& -Q_{H}+W-Q_{C},\label{First_Law}
\end{eqnarray}
where $E(\cdot)=\Tr(H_S\;\cdot)$ denotes average energy, and in the last line, we have used Eqs. \eqref{Defn:HeatExchanged}, \eqref{Exp_Work}, and  \eqref{cold_heat}. Equation \eqref{First_Law} establishes the \emph{first law of thermodynamics} as the manifestation of the conservation of energy 
\begin{equation}
    W = Q_H+Q_C,
\end{equation}
which allows us to express the work done by the engine as the net average flow of heat. Then, the efficiency of the engine is defined as 
\begin{equation}
    \eta := \frac{W}{Q_H} = \frac{Q_C+Q_H}{Q_H} = 1+\frac{Q_C}{Q_H}.
\end{equation}
Next, we demonstrate that the efficiency of the engine is upper bounded by the Carnot efficiency, therefore establishing the second law of thermodynamics. The result is stated in the following theorem.
\begin{thm}
    The efficiency of the three-stoke engine is bounded by the Carnot efficiency, 
    \begin{equation}
        \eta \leq 1-\frac{\beta_H}{\beta_C}:=\eta_{
        \textnormal{Carnot}}.
    \end{equation}
\end{thm}
The proof of the theorem can be found in appendix \ref{Thermo_framework_carnot}. This establishes the thermodynamic framework of the three-stroke heat engine. We shall now optimize the work production and efficiency of the three-stroke engine.

\subsection{Optimal performance of the three stroke engine}

In this section, we characterize the optimal efficiency and work produced by the three-stroke engine.  

\begin{thm}\label{thm_optimal_work_eff}
Consider a three-stroke engine with a two-level system as the working body with Hamiltonian $H_S=\omega|1\rangle\langle 1|$ and the range of $\lambda$ in Eq. \eqref{Range_of_lambda_gen} is restricted between $0$ to $\lambda^{\textnormal{max}}_H$ and $\lambda^{\textnormal{max}}_C$ for the first and third strokes, respectively. Then, the optimal work produced by the engine and the optimal efficiency are given by 
\begin{widetext}
\begin{eqnarray}
p &=& \frac{1-\lambda^{\textnormal{max}}_H(1-\lambda_C^{\textnormal{max}})-\lambda_C^{\textnormal{max}}(1-\lambda^{\textnormal{max}}_H)e^{-\beta_C\omega}}{2-\lambda_C^{\textnormal{max}}(1-\lambda_H^{\textnormal{max}})-\lambda_H^{\textnormal{max}}-\lambda_C^{\textnormal{max}}(1-\lambda_H^{\textnormal{max}})e^{-\beta_C\omega}-\lambda_H^{\textnormal{max}}(1-\lambda_C^{\textnormal{max}})e^{-\beta_H\omega}+\lambda_H^{\textnormal{max}}\lambda_C^{\textnormal{max}}e^{-(\beta_C+\beta_H)\omega}},\label{Opt_performance_p}\\
W_{\textnormal{max}} &=& 1-2\lambda^{\textnormal{max}}_H+\frac{2\Big(\lambda_H^{\textnormal{max}}e^{-\beta_H\omega}-(1-\lambda_H^{\textnormal{max}})\Big)\Big(1-(1-\lambda_C^{\textnormal{max}})\lambda_H^{\textnormal{max}}-\lambda_C^{\textnormal{max}}e^{-\beta_C\omega}(1-\lambda_{H}^{\textnormal{max}})\Big)}{2-\lambda_C^{\textnormal{max}}(1-\lambda_H^{\textnormal{max}})-\lambda_H^{\textnormal{max}}-\lambda_C^{\textnormal{max}}e^{-\beta_C\omega}(1-\lambda_H^{\textnormal{max}})-\lambda_H^{\textnormal{max}}e^{-\beta_H\omega}(1-\lambda_C^{\textnormal{max}})+\lambda_{C}^{\textnormal{max}}\lambda_{H}^{\textnormal{max}}e^{-(\beta_C+\beta_H)\omega}},\nonumber\\\label{Opt_performance_W}\\
\eta_{\textnormal{max}} &=& 1-\frac{\lambda_C^{\textnormal{max}}\Big(1-\lambda_H^{\textnormal{max}}e^{-\beta_H\omega}-(1-\lambda_H^{\textnormal{max}})e^{-\beta_C\omega}\Big)}{\lambda^{\textnormal{max}}_H\Big(e^{-\beta_H\omega}-(1-\lambda^{\textnormal{max}}_C)-\lambda_C^{\textnormal{max}}e^{-(\beta_H+\beta_C)\omega)}\Big)}.
\label{Opt_performance_E}
\end{eqnarray}
\end{widetext}
\end{thm}
The complete proof of this theorem can be found in Appendix \ref{Proof_of_thm} which requires the concept of $\beta$-ordering that we have described in Appendix \ref{beta_order_section} \cite{horodecki2013fundamental,Lostaglio_2019}. Here, we would like to make an important remark. The expression of the optimal efficiency $\eta_{\textnormal{max}}$ in Eq. \eqref{Opt_performance_E} reduces to $1+\frac{\lambda_{C}^{\textnormal{max}}}{\lambda_{H}^{\textnormal{max}}\left(1-\lambda_{C}^{\textnormal{max}}\right)}$
when $\omega\rightarrow\infty$, which leads to  efficiency more than 1. However, in this regime $W_{\text{max}}$  becomes negative which can be checked by taking the limit $\omega\rightarrow\infty$ in Eq. \eqref{Opt_performance_W}. Note that from Lemma \ref{Lemworkheat} in Appendix \ref{Thermo_framework_carnot}, one can see that in order to produce positive amount of work the efficiency of the engine should be in between $0$ and $\eta_{\text{Carnot}}=1-\frac{\beta_H}{\beta_C}$ i.e.,
\begin{equation}\label{Workefficiencyimplication}
    W>0 \Rightarrow 0 <\eta <\eta_{\text{Carnot}}.
\end{equation}
By taking negation of Eq. \eqref{Workefficiencyimplication}, we see the following:
\begin{equation}
    \text{if $\eta \leq 0$ or $\eta \geq \eta_{\text{Carnot}}$  then }\,\, W\leq 0.
\end{equation}
This implies that the engine fails to produce a positive amount of work when the efficiency of the engine is not in between $0$ and $\eta_{\text{Carnot}}$. Therefore, for a given value of $\lambda_H^{\textnormal{max}} \in[0,1]$, $\lambda_C^{\textnormal{max}}\in[0,1]$, $\beta_H$ and $\beta_C$ (with the obvious condition $\beta_C>\beta_H$), we need to suitably choose $\omega$ so that $\eta_{\text{max}}$ lies between $0$ and $\eta_{\text{Carnot}}$ to produce a positive amount of work. 
 
Let us look at a scenario where one can perform any thermal operation at inverse temperatures $\beta_H$ and $\beta_C$ in the first and third strokes, respectively. This implies that one can choose $\lambda^{\text{max}}_{H} = 1$ and $\lambda^{\text{max}}_{C} = 1$. With these values, one can use the general expressions in Eqs. \eqref{Opt_performance_W} and \eqref{Opt_performance_E} to find the expressions for the engine's maximum work production and maximum efficiency. In that scenario, the optimal work produced by the engine and the optimal efficiency from Eqs. \eqref{Opt_performance_W} and \eqref{Opt_performance_E} are reduced to \cite{Lobejko2020}
\begin{eqnarray}
W_{\text{max}} &=& \omega\Big(\frac{2e^{-\beta_H\omega}}{1+e^{-(\beta_H+\beta_C)\omega}}-1\Big)\,,\mbox{and}\nonumber\\
\eta_{\text{max}} &=& 1-\frac{1-e^{-\beta_H\omega}}{e^{-\beta_H\omega}-e^{-(\beta_H+\beta_C)\omega}},
\end{eqnarray}
respectively.  Note that, in order to produce a positive amount of work (i.e., $W_{\text{max}}>0$), the necessary and sufficient condition must be 
\begin{equation}
    2>e^{\beta_H\omega}+e^{-\beta_C\omega}.
\end{equation} 

Therefore, one must follow the protocol in Box \ref{box:one} to achieve the best work production and efficiency of the stroke-based discrete heat engine with the hot and cold heat baths at inverse temperatures $\beta_H$ and $\beta_C$, respectively. This specific model of the three-stroke engine has been extensively studied in Ref. \cite{Lobejko2020}. In contrast, our results, which encompass the maximum work production and maximum efficiency given in Eqs. \eqref{Opt_performance_W} and \eqref{Opt_performance_E}, have broader applicability.

 \begin{frameenv}{Protocol leading to optimal work production and efficiency}
    \label{box:one}
\begin{enumerate}
    \item Take a two level system in the state 
    \begin{equation}\label{cyclic_state2}
        \rho = \begin{pmatrix}
       \frac{1}{1+e^{-(\beta_H+\beta_C)\omega}} & 0\\
       0 & \frac{e^{-(\beta_H+\beta_C)\omega}}{1+e^{-(\beta_H+\beta_C)\omega}}
    \end{pmatrix}
    \end{equation}
\item Apply a thermal operation that corresponds to the Gibbs stochastic matrix 
\begin{equation}
        T^{\E_{H}^*} = \begin{pmatrix}
       1-e^{-\beta_H\omega} & 1\\
       e^{-\beta_H\omega} & 0
    \end{pmatrix}
    \end{equation}
wrt. the hot heat bath. 
\item Apply the permutation 
\begin{equation}
        X = \begin{pmatrix}
       0 & 1\\
       1 & 0
    \end{pmatrix}
    \end{equation}
on the system to extract work from the system and transfer it to the battery.
\item Apply a thermal operation that corresponds to the Gibbs stochastic matrix 
\begin{equation}
        T^{\E_{C}^*} = \begin{pmatrix}
       1-e^{-\beta_C\omega} & 1\\
       e^{-\beta_C\omega} & 0
    \end{pmatrix}
    \end{equation}
wrt. the cold heat bath which closes the cycle. 
\end{enumerate}
\end{frameenv}

From Eq. \eqref{Opt_performance_E}, one can also recover the efficiency of the open cycle engine that we have described in subsection \ref{Open_cycle_engine}. Since, in the open cycle engine, the working body starts with a thermal state at the temperature of the cold bath $\beta_C$,  closing the cycle during the final stroke reduces to re-thermalization at the inverse temperature $\beta_C$. Therefore, by substituting
\begin{equation}\label{Limit_thermal_open_cycle}
    \lambda_{H}^{\text{max}}=1\,\,\textrm{and}\,\, \lambda_{C}^{\text{max}} = \frac{1}{1+e^{-\beta_C\omega}}
\end{equation}
in Eqs. \eqref{Opt_performance_W} and \eqref{Opt_performance_E}, we obtain \cite{Biswas2022extractionof}: 
\begin{eqnarray}
    &&W_{\text{max}} = \omega\Big(\frac{2e^{-\beta_H\omega}}{1+e^{-\beta_C\omega}}-1\Big),\,\,\mbox{and}\\
    &&\eta_{\text{max}} = 1-\frac{1-e^{-\beta_H\omega}}{e^{-\beta_H\omega}-e^{-\beta_C\omega}}.
\end{eqnarray}

Next, we shall focus on characterizing the efficiency of the three-stroke engine when we do not have access to a full set of thermal operations. Namely, we focus on two such scenarios: first, we consider the case where we study the optimal performance of the engine when we restrict ourselves to thermal operations that can be realized via the Jaynes-Cummings interaction, and second, when we restrict ourselves to thermal operations that can be performed when one has access only to a finite-dimensional heat bath that can be modeled by a truncated harmonic oscillator. 

\begin{figure}[t]
    \centering
    \centering
    \includegraphics[width=
    \columnwidth]{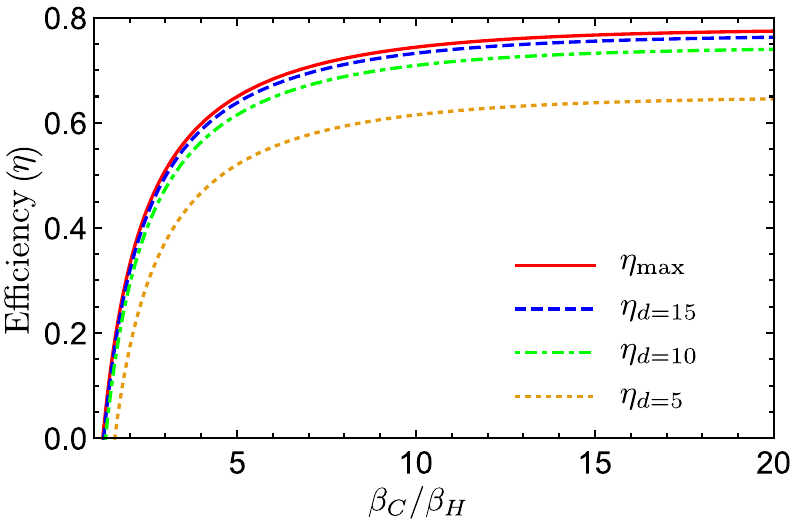}
    \caption{Plot of efficiency ($\eta$) vs $\beta_C/\beta_H$ with $\omega=0.2/\beta_H$. The solid, dashed, dotdashed, and dotted curves represent optimal efficiency, efficiencies for finite heat bath with $d=15$, $d=10$, and $d=5$, respectively.}
    \label{fig:efficiency_finiteheatbath}
\end{figure}

\begin{figure}[t]
    \centering
    \centering
    \includegraphics[width=
    \columnwidth]{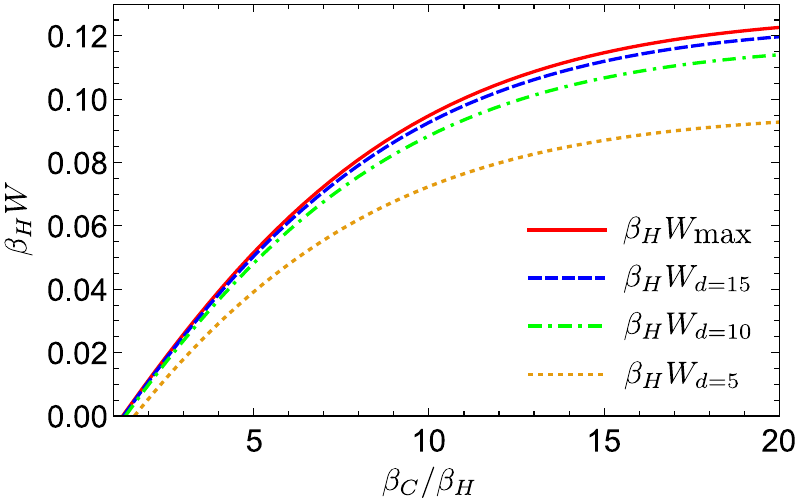}
    \caption{Plot of $\beta_H W$ vs $\beta_C/\beta_H$ with $\omega=0.2/\beta_H$, and $W$ represents the work done by the engine. The solid, dashed, dotdashed, and dotted curves represent optimal work production, work productions for finite heat bath with $d=15$, $d=10$, and $d=5$, respectively.  }
    \label{fig:work_finiteheatbath}
\end{figure}

\subsubsection{Optimal performance of the engine for finite sized heat baths}

In this section, we explore the engine's optimal efficiency and work output when we only have access to finite-dimensional heat baths that can be modeled using a truncated harmonic oscillator. In this scenario, one cannot perform any arbitrary thermal operation in the first and third strokes, which implies $\lambda^{\text{max}}<1$ (See Eq. \eqref{Range_of_lambda_gen}). We need to figure out the value of $\lambda^{\text{max}}$ when we have access only to the finite-dimensional heat baths modeled by a truncated harmonic oscillator. This is addressed in the following theorem.

\begin{thm}\label{Finite_size_bath_range}
    Consider a thermal operation $\mathcal{E}$ on a two-level system given in the state $\rho$ 
    \begin{equation}
        \mathcal{E}(\rho)=\Tr_{E}\Big(U(\rho\otimes \gamma_E)U^{\dagger}\Big),
    \end{equation}
    where $\rho$ is a energy incoherent state, i.e., $[\rho,H_S]=0$, and $[U,H_S+H_E]=0$. Let us model the Hamiltonian of the system as
    \begin{equation}
        H_S = \omega|1\rangle\langle 1|,
    \end{equation}
    which is resonant with the bath Hamiltonian modeled by a truncated harmonic oscillator, i.e.,
    \begin{equation}\label{truncated_bath}
        H_E = \omega \sum_{n=0}^d n|n\rangle\langle n|,
    \end{equation}
    where $d+1$ is the dimension of the heat bath. Then, the maximum value of $\lambda$ in Eq. \eqref{Range_of_lambda} is given by 
    \begin{equation}\label{FSE_range}
        \lambda^{\textnormal{max}} = \frac{1-e^{-\beta\omega d}}{1-e^{-\beta\omega(d+1)}}.
    \end{equation}
\end{thm}

 The proof of this theorem can be found in Appendix \ref{Appendix_finite_sized_baths}. One can clearly see from the above Eq. \eqref{FSE_range} that when $d\rightarrow\infty$, then $\lambda^{\text{max}}\rightarrow 1$, which means any thermal operation on a two-level system can be realized with the bath Hamiltonian modeled by a harmonic oscillator. Another interesting limit is when $d\rightarrow 1$, which implies the heat bath is modeled by a two-level system (See Eq. \eqref{truncated_bath}). In that case, we have
 \begin{equation}\label{above47}
     \lambda^{\text{max}} = \frac{1}{1+e^{-\beta\omega}},
 \end{equation}
 and substituting Eq. \eqref{above47} into Eq. \eqref{convex_hull}, we obtain 
 \begin{equation}
     \mathcal{T}(\rho) = \text{Conv}\Big\{ \rho, \frac{1}{1+e^{-\beta\omega}}\begin{pmatrix}
         1 & 0\\
         0 & e^{-\beta\omega}
     \end{pmatrix}\Big\}.
 \end{equation}

This suggests that when a two-level system functions as a heat bath, it restricts thermal operations between an identity map and thermalization. As a result, when we model a two-level system resonant with the working body as a heat bath, the three stroke engine cannot extract work since the heat stroke fails to induce non-passivity in the working body. Observe that Eq. \eqref{Limit_thermal_open_cycle} resembles Eq. \eqref{above47}. This suggests that the process of re-thermalizing the state of the two-level working body for an open-cycle engine (discussed in Sec. \ref{Open_cycle_engine}) is identical to the proposed three-stroke engine that interacts with a cold heat bath modeled using a truncated harmonic oscillator with two levels (this can be observed by substituting $d=1$ in Eq. \eqref{truncated_bath}).

Employing the result from theorem \ref{thm_range_thermal} of Appendix \ref{Appendix_finite_sized_baths}, the values of $\lambda^\textrm{max}$ for the inverse temperatures $\beta_C$ and $\beta_H$ reduce to 
\begin{equation}\label{lambda_finite_d}
    \lambda^{\text{max}}_{C,d} := \frac{1-e^{-\beta_C\omega d}}{1-e^{-\beta_C\omega (d+1)}},\,\,\mbox{and}\,\,\lambda^{\text{max}}_{H,d} := \frac{1-e^{-\beta_H\omega d}}{1-e^{-\beta_H\omega (d+1)}},
\end{equation}
respectively. Substituting these values in Eq. \eqref{Opt_performance_E}, we  obtain 
\begin{equation}
    \eta_d = 1-\frac{(1-e^{-d\beta_C\omega})(1-e^{-\beta_H\omega})(1-e^{-(\beta_C+d\beta_H)\omega})}{(1-e^{-\beta_C\omega})(e^{-\beta_H\omega}-e^{-d\beta_C\omega})(1-e^{-d\beta_H\omega})}.
\end{equation}
Similarly, in this scenario, one can obtain an expression for optimal work production ($W_d$) by substituting the values of $\lambda$s given in Eq. (\ref{lambda_finite_d}) into Eq. (\ref{Opt_performance_W}). We plot the work production and efficiency in Figs. \ref{fig:efficiency_finiteheatbath} and \ref{fig:work_finiteheatbath} for $\omega=0.2/\beta_H$ and some $d$ values, and compared them with $W_{\max}$ and $\eta_{\max}$, respectively. 

\subsubsection{Optimal performance of the engine modelled with Jaynes-Cummings interaction}

In this section, we explore the Jaynes-Cummings model and determine the optimal work and efficiency for the stroke-based discrete engine we considered. Let us recall that the Hamiltonian of the system is $H_S=\omega\ketbra{1}{1}$. The Hamiltonian of the bath is modelled by a single mode harmonic oscillator, i.e.,
\begin{equation}
    H_E = \sum_{n=0}^{\infty}n\omega\ketbra{n}{n},
\end{equation}
where $\ket{n}$ denotes the eigenstate of $H_E$ with eigenvalue $n\omega$. The interaction between the working body and the bath is described by a resonant Jaynes-Cummings Hamiltonian $H_\text{JC}$ in rotating wave approximation as follows:
\begin{equation}
    H_\text{JC} = g(\sigma_+\otimes a+\sigma_-\otimes a^{\dagger}),
\end{equation}
where $a^{\dagger}$ and $a$ represent the creation and annihilation operators on the bath, while $\sigma_{+}=\ketbra{1}{0}$ and $\sigma_- = \ketbra{0}{1}$ represent the excitation and de-excitation operators on the working body, respectively. Clearly $[H_S+H_E,H_\text{JC}]=0$, therefore the transformation
\begin{eqnarray}
    \rho_S &\rightarrow& \Tr_E\left(e^{-it H_\text{JC}}\left(\rho_S\otimes \frac{e^{-\beta_H H_E}}{\Tr\left(e^{-\beta_H H_E}\right)}\right)e^{it H_\text{JC}}\right)\nonumber\\&:=&\E_{H,\text{JC}}\left(\rho_S\right),
\end{eqnarray}
is a thermal operation at the inverse temperature $\beta_H$, which we denoted by $\mathcal{E}_{H,\text{JC}}$. 
In a similar manner, starting from the initial state $V\mathcal{E}_{H,\text{JC}}(\rho_S)V^{\dagger}$, one can describe the thermal operation $\mathcal{E}_{C,\text{JC}}$ in the third stroke at the inverse temperature $\beta_C$ (recall that from Eq. \eqref{third_stroke_eqn}, thermal operation $\E_{C,\text{JC}}$ in the third stroke acts on the state that results after the second stroke, i.e., $V\mathcal{E}_{H,\text{JC}}(\rho_S)V^{\dagger}$). The maximum achievable values of $\lambda^{\text{max}}_{H}$ and $\lambda^{\text{max}}_{C}$ (denoted by $\lambda^{\text{max}}_{H,\text{JC}}$ and $\lambda^{\text{max}}_{C,\text{JC}}$) for the thermal operations $\E_{H,\text{JC}}$ and $\E_{C,\text{JC}}$ in the first and third strokes when it is modelled by Jaynes Cummings interaction are as follows: (see Eq. 10 of Ref. \cite{Lostaglio2018elementarythermal}) 
\begin{widetext}
    \begin{equation}\label{case:1}
    \lambda^{\textrm{max}}_{H,\text{JC}}=\begin{cases}
   \frac{1}{16}\left(8 e^{-\beta_H\omega}-e^{-2\beta_H\omega}+e^{3\beta_H\omega}+8\right), & \text{for } 0 \leq \beta_H\omega \leq \frac{\log 4}{3}, \\
   e^{-4\beta_H\omega}-e^{-3\beta_H\omega}+1 , & \text{for } \beta_H\omega \geq \frac{\log 4}{3},
  \end{cases}
\end{equation}
and 
\begin{equation}\label{case:2}
    \lambda^{\textrm{max}}_{C,\text{JC}}=\begin{cases}
   \frac{1}{16}\left(8 e^{-\beta_C\omega}-e^{-2\beta_C\omega}+e^{3\beta_C\omega}+8\right), & \text{for } 0 \leq \beta_C\omega \leq \frac{\log 4}{3}, \\
   e^{-4\beta_C\omega}-e^{-3\beta_C\omega}+1 , & \text{for } \beta_C\omega \geq \frac{\log 4}{3}.
  \end{cases}
\end{equation}
\end{widetext}

\begin{figure}[t]
   \centering    \includegraphics[width=
    \columnwidth]{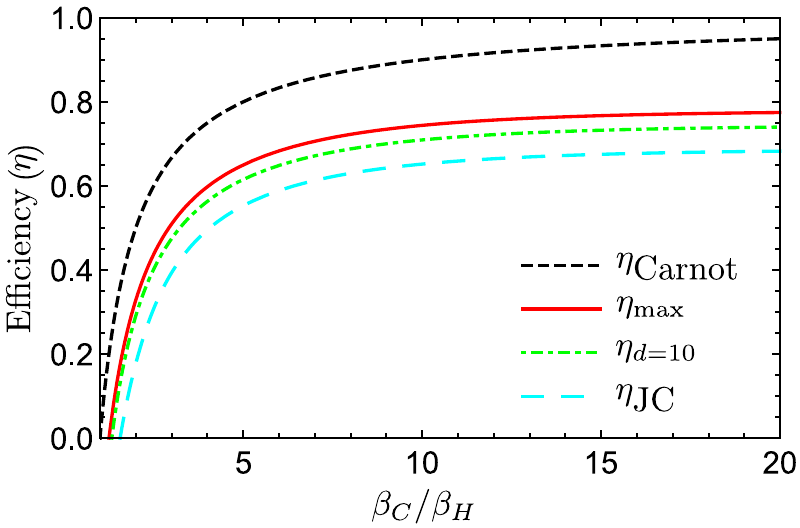}
    \caption{Plot of efficiency ($\eta$) vs $\beta_C/\beta_H$ with $\omega=0.2/\beta_H$. The dashed, solid, dotdashed, and big-dashed curves represent Carnot efficiency, optimal efficiency, efficiency for finite heat bath with $d=10$, and efficiency for Jaynes-Cummings model, respectively.  }
    \label{fig:efficiency_JC}
\end{figure}

\begin{figure}[t]
    \centering    \includegraphics[width=
    \columnwidth]{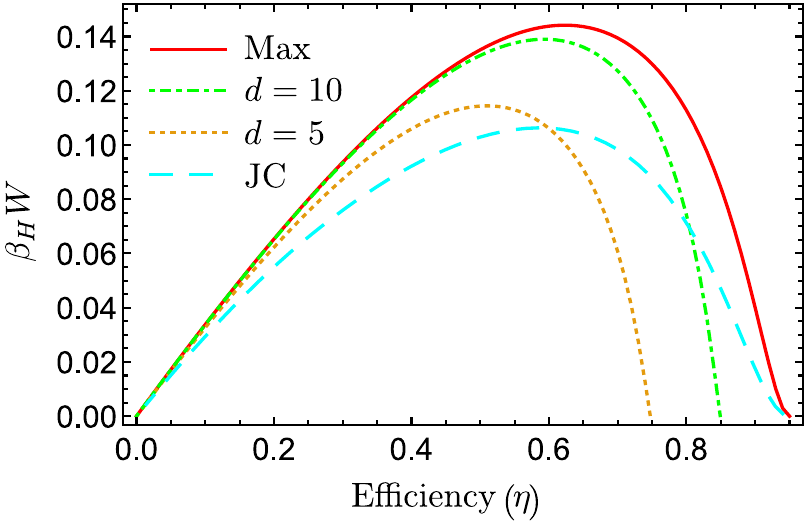}
    \caption{Plot of $\beta_H W$ vs efficiency ($\eta$), where $W$ represents work done by the engine. The solid, dotdashed, dotted, and big-dashed curves represent optimal work, work for finite heat bath with $d=10$, $d=5$, and work for Jaynes-Cummings, respectively.}
    \label{fig:Work_JC}
\end{figure}

Now, substituting  $\lambda^{\textrm{max}}_{H,\text{JC}}$ and $\lambda^{\textrm{max}}_{C,\text{JC}}$ from Eqs. \eqref{case:1} and \eqref{case:2} into Eqs. \eqref{Opt_performance_W} and \eqref{Opt_performance_E}, we can obtain the optimal values of work and efficiency for the engine. The resulting expressions are cumbersome, hence we opted not to put them here; rather, we compare the efficiency obtained in this case with the other cases as shown in Fig. \ref{fig:efficiency_JC}. Furthermore, we also make a comparison of the tradeoff between work and efficiency in Fig. \ref{fig:Work_JC} for various scenarios we have discussed so far.

\section{Conclusion}\label{conclusion}

In this paper, we studied a stroke-based engine in the microscopic regime that requires three strokes to complete the thermodynamic cycles. In particular, we considered the three-stroke engine that involves a two-level system as a working body that steers the energy flow from a hot heat bath to a cold heat bath. The transformation of the working body in the presence of both heat baths is described by thermal operations. We figured out the optimal protocol that leads to the optimal production of work as well as the optimal efficiency when one can perform a restricted class of thermal operations on the working body in the presence of hot and cold heat baths. We analyzed the optimal efficiency and work production by the engine for two such restricted classes: first, when one has access to finite-sized heat baths, and second, when the thermal operations involved in the first and third strokes are realized via the Jaynes-Cummings interaction. Furthermore, we showed that the obtained expressions of optimal work production and efficiency reduce to the expressions of optimal efficiency and work production given in Theorem 3 of Ref. \cite{Lobejko2020} when we relax the restrictions, hence allowing any arbitrary thermal operation on the working body in the presence of hot and cold baths.    

Our research paves the way for numerous exciting research directions. One could explore how the optimal efficiency and maximum work produced by the engine change with the dimension of the working body. Exploring how the engine's optimal efficiency and maximum work production vary with the dimensions of the working body presents an intriguing challenge. The complexity of this problem lies in the intricate structure of the thermal cone, which becomes increasingly convoluted as the dimension of the working body grows. On top of that, the complexity inherent in the Hamiltonian structure of the working body further adds to the intricacy of the problem.

Investigating the significance of coherence in the operation of three-stroke engines presents a captivating avenue for research.  In this paper, we have focused only on the working bodies that do not have any coherence in energy eigenbasis (i.e., in the eigenasis of Hamiltonian $H_S$). Ref. \cite{Lobejko2020} demonstrated that coherence in the state of the working body, modeled as a two-level system, fails to improve the performance of a three-stroke engine.  This raises the question: can this result be generalized for arbitrary dimensional working body? Specifically, it will be interesting to prove or disprove that the coherence present in the state of the working body is not useful to improve the performance of a three-stroke engine. However, characterizing the transformation of coherence in a state under thermal operations remains a formidable challenge, particularly beyond two dimensions. Nonetheless, related results suggest that the dephasing of the state in the local energy eigenbasis is compatible with thermal operations in all dimensions \cite{LostaglioKorzekwaPRX,Lostaglio_2019}. These results might be helpful in this endeavor. 

This work centers on stroke-based heat engines, where work extraction and heat exchange happen in discrete, distinct time intervals. We anticipate a correspondence between the three-stroke discrete heat engine and the continuous heat engine. Notably, analyzing the continuous heat engine enables us to compute the engine's power or rate of work extraction with time, while the proposed three-stroke heat engine yields the average work per cycle. Upon rigorously bridging the gap between these stroke-based discrete and continuous heat engines, an interesting question emerges: what is the minimum time required to complete a thermodynamic cycle? We envision that elucidating the connection between the three-stroke heat engine and the continuous heat engine could offer valuable insights into realizing the three-stroke heat engines experimentally. In a related study, as proposed in Ref. \cite{Yu19}, an experimental implementation of an autonomous refrigerator within the framework of cavity quantum electrodynamics was suggested. Moreover, recent work outlined in Ref. \cite{NYHSG} explores the construction of an autonomous refrigerator using superconducting circuits. Therefore, we believe that designing the three-stroke heat engine is attainable within the frameworks of cavity QED and superconducting circuits.

\section*{Acknowledgement}
TB acknowledges Pawe\l{ }Mazurek for insightful discussions and comments.  The part of the work done at Los Alamos National Laboratory (LANL) was carried out under the auspices of the US DOE and NNSA under contract No.~DEAC52-06NA25396. TB also acknowledges support by the DOE Office of Science, Office of Advanced Scientific Computing Research, Accelerated Research for Quantum Computing program, Fundamental Algorithmic Research for Quantum Computing (FAR-QC) project. TB also acknowledge support from the Foundation for Polish Science through IRAP project co-financed by EU within the Smart Growth Operational Programme (contract no.2018/MAB/5). CD acknowledges support from the German Federal
Ministry of Education and Research (BMBF) within the funding program ``quantum technologies -- from basic research to market'' in the joint project QSolid (grant
number 13N16163).

\appendix
\onecolumngrid
\section{Upper bounding the efficiency of the three stroke engine using Carnot efficiency}\label{Thermo_framework_carnot}
\begin{lem}\label{Lemworkheat}
For any discrete stroke-based heat engine, if the work production $W$ is positive, then the amount of heat flowing from the hot heat bath $Q_H$ is also positive, i.e.,
\begin{equation}
    W> 0 \Rightarrow Q_H>0.
\end{equation}
Furthermore, the efficiency of the discrete stroke-based engine will be upper bounded by Carnot efficiency, i.e.,
\begin{equation}
    \eta=\frac{W}{Q_H}\leq 1-\frac{\beta_H}{\beta_C}=\eta_\textnormal{Carnot}.
\end{equation}
\end{lem}

\begin{proof}
    From Eq. \eqref{Exp_Work} of section \ref{Characterization}, we know that 
    \begin{equation}
    W = \Tr\Big(H_S\big(\mathcal{E}_{H}(\rho_S)- V (\mathcal{E}_{H}(\rho_S)) V^{\dagger}\big)\Big).
    \end{equation}
 and we rewrite the transfer from the hot and cold baths as Eq. \eqref{Defn:HeatExchanged} and cold heat Eq. \eqref{cold_heat}

\begin{eqnarray}
&& Q_H=\Tr\Big(\big(H_S\big(\mathcal{E}_{H}(\rho_S)-\rho_S\big)\Big)\,\, \mbox{and}\\
&& Q_C = \Tr\Big( H_S\big(\E_C(V\E_H(\rho_S)V^{\dagger})-V\E_H(\rho_S)V^{\dagger}\big)\Big)\nonumber\\
&&\,\,\,\,\,\,\,\,\,\,=\Tr\Big( H_S\big(\rho_S-V\E_H(\rho_S)V^{\dagger}\big)\Big)\label{QC21}
\end{eqnarray}
respectively, where the second equality in Eq. \eqref{QC21} follows from the cyclicity condition, i.e., 
\begin{equation}\label{cyc1}
    \E_C(V\E_H(\rho_S)V^{\dagger}) = \rho_S
\end{equation}
Therefore, we have 
\begin{equation}\label{total_heat1}
 Q_H+Q_C>0.
\end{equation}
Next, using the monotonicity of free energy under thermal operations, we can write the following:
\begin{equation}\label{Eqfreeenergy1}
    F(\rho_S)-F(\mathcal{E}_H(\rho_S))\geq 0, 
\end{equation}
where $F(\cdot)=\Tr(H_S\cdot)-S(\cdot)/\beta_H$ where $S(\cdot)$ denotes the von Neumann entropy. Thus, one can simplify Eq. \eqref{Eqfreeenergy1} as
\begin{eqnarray}
&&\Tr(H_S\rho_S)-\Tr\big(H_S\mathcal{E}_H(\rho_S)\big) -\frac{1}{\beta_H}\big(S(\rho_S)-S(\mathcal{E}_H(\rho_S))\big) \geq 0\nonumber\\&\Rightarrow& -\frac{1}{\beta_H}\big(S(\rho_S)-S(\mathcal{E}_H(\rho_S))\big)\geq \Tr(H_S\rho_S)-\Tr\big(H_S\mathcal{E}_H(\rho_S)\big) = Q_H\nonumber\\
&\Rightarrow&
   -\Big( S(\rho_S)-S(\mathcal{E}_H(\rho_S))\Big)\geq \beta_H Q_H.\label{entropy_hot1}
\end{eqnarray}
Similarly, for the cold heat bath, we find
\begin{eqnarray}\label{entropy_cold1}
   &&-\Big( S(V\mathcal{E}_H(\rho_S)V^{\dagger})-S\big(\mathcal{E}_C(V\mathcal{E}_H(\rho_S)V^{\dagger})\big)\Big)\geq \beta_C Q_C\nonumber\\
   &\Rightarrow&-\Big( S(\mathcal{E}_H(\rho_S))-S(\rho_S)\Big)\geq \beta_C Q_C\label{QCF1},
\end{eqnarray}
where we have use the fact $S(V\mathcal{E}_H(\rho_S)V^{\dagger})=S(\mathcal{E}_H(\rho_S))$ as von-Neumann entropy is invariant under unitary operation, and cyclicity condition given in Eq. \eqref{cyc1} to write Eq. \eqref{QCF1}.
Now, adding Eqs. (\ref{entropy_hot1}) and (\ref{entropy_cold1}), we obtain $\beta_HQ_H+\beta_CQ_C\leq0$, that gives 
\begin{equation}\label{upbound1}
    -\frac{\beta_H}{\beta_C}Q_H\geq Q_C.
\end{equation}
Substituting this upper bound of $Q_C$ given in Eq. \eqref{upbound1} into Eq. \eqref{total_heat1} and noting the fact that $\beta_H<\beta_C$, we arrive at
\begin{eqnarray}
Q_H\Big(1-\frac{\beta_H}{\beta_C}\Big) > 0\quad\Rightarrow \;Q_H>0,
\end{eqnarray}
which completes the proof of the first part. As a consequence, employing Eq. \eqref{upbound1} we obtain
\begin{equation}
    \eta = \frac{W}{Q_H}=\frac{Q_H+Q_C}{Q_H} = 1+\frac{Q_C}{Q_H}\leq 1-\frac{\beta_H}{\beta_C},
\end{equation}
which completes the proof for the second part.
\end{proof}

\section{Thermal operation and $\beta$-ordering}\label{beta_order_section}

We have introduced a thermal operation on a given state $\rho_S$  of dimension $d$
as
\begin{equation}\label{eq:TO_Again}
    \E(\rho_S)=\Tr_{E}\left({U\left(\rho_S\otimes\gamma_E\right)U^{\dagger}}\right),
\end{equation}
where $U$ is a joint unitary  operation  acting on the system and a heat bath that commutes with the total Hamiltonian of the system and bath, i.e., $[U, H_S+ H_E] = 0$, and $\gamma_{E}$ is a thermal state of the bath with very large volume at some fixed inverse temperature $\beta$. If $\rho_S$ is an energy-incoherent state, i.e., $[\rho_S,H_S] =0$, one can prove 
\cite{horodecki2013fundamental}
\begin{equation}
    \rho_S\xrightarrow{\E} \E(\rho_S):=\sigma_S\quad\Longrightarrow\quad T^{\E}(\v p) = \v q,\quad T^{\E}(\v \gamma) = \v \gamma,
\end{equation}
where $\v p$ , $\v q$ and $\v \gamma$ are the population vectors in $d$-dimension constructed from $\rho_S$ , $\E(\rho_S)$ and the Gibbs state of the system $\gamma$, respectively. Recall that we defined a population vector for any state of the system $\rho$ as a probability vector with entries $p_i = \bra{E_i}\rho\ket{E_i}$, where $\{\ket{E_i}\}$s correspond to the eigenbasis of the Hamiltonian of the system.

Employing the concepts of relative majorization, one can verify that the existence of such a Gibbs-stochastic matrix $T^{\mathcal{E}}$ is equivalent to the thermomajorization condition $\v p \succ_{\gamma} \v q$. For a detailed review of thermomajorization, we refer to Ref. \cite{Lostaglio_2019, Entropy_sagawa, Marshall1980InequalitiesTO}. Here, we will outline the ideas of the thermonajorization condition using the concept of the thermomajorization curve, which is necessary to analyze the engine's performance.

Let $L_{\v{\gamma}}(\v{p})$ be a piecewise linear curve in $\mathbb{R}^2$ obtained by joining $\Big(Z_S\sum_{i=1}^k \gamma_{\pi(i)}, \sum_{i=1}^k p_{\pi(i)}\Big)$ for $k = \;1, \ldots, d$. Here $\pi$ is the permutation acting on the set of $d$ indices such that 
\begin{eqnarray}
  \frac{p_{\pi(1)}}{\gamma_{\pi(1)}} \geq \frac{p_{\pi(2)}}{\gamma_{\pi(2)}} \ldots \geq\frac{p_{\pi(d)}}{\gamma_{\pi(d)}}, 
\end{eqnarray}
and $Z_S$ is the partition function of the system obtained from the Hamiltonian  $H_S$. 
The curve $L_{\v{\gamma}}(\v{p})$ is referred to as the \emph{thermomajorization curve} of $\v p$ with respect to $\v{\gamma}$, and the ordering obtained by the action of permutation $\pi$ on the set of indices $(12\ldots d)$, i.e.,
\begin{equation}
    \pi(12\ldots d), 
\end{equation}
 is referred to as \emph{$\beta$-order} of $\v p$ with respect to $\v{\gamma}$. We say $\v p$ thermomajorizes $\v q$ and denote it as $\v p\succ_{\gamma}\v q$ if and only if the thermomajorization curve $L_{\v{\gamma}}(\v{p})$ does not lie below the thermomajorization curve $L_{\v\gamma}(\v q)$.  Therefore, the existence of a thermal operation that transforms $\rho$ to $\sigma$ is equivalent to the existence of a Gibbs-stochastic matrix that transforms the population vector $\v p$ associated with $\rho$, to the population vector $\v q$ associated with $\sigma$, and this further reduces to the condition $\v p\succ_{\gamma}\v q$. In the limit $\beta\rightarrow 0$, $\v p\succ_{\gamma}\v q$ boils down to the well-known majorization condition $\v p\succ\v q$.

\section{Proof of theorem \ref{thm_optimal_work_eff}}\label{Proof_of_thm}
We begin by re-writing the population vector associated with the initial state $\rho_S$ as
\begin{equation}
    \v p = \begin{pmatrix}
        p & 1-p
    \end{pmatrix}^T.
\end{equation}
Since we are working with a two-level system as a working body, it can have only two possible $\beta$-order wrt. inverse temperature $\beta_H$: 
\begin{eqnarray}
    &&(12) \quad\text{when}\quad p\geq \frac{1-p}{e^{-\beta_H\omega}}\label{C2},\\
    &&(21) \quad\text{when}\quad p\leq \frac{1-p}{e^{-\beta_H\omega}}.
\end{eqnarray}

If the population vector of the initial state of the working body has $\beta$-order $(21)$, then it has the highest average energy among all states that can be achieved from the initial state $\rho_S$ via thermal operation at inverse temperature $\beta$  (See Lemma 16 of Ref. \cite{Biswas2022extractionof}). Therefore, if the population vector of the initial state of the working body has $\beta$-order $(21)$, then the change in the average energy under any thermal operation $\E_{\beta}(\cdot)$ at inverse temperature $\beta$ will either be zero or negative 
\begin{equation}
    \Tr(H_S\E_{H}(\rho_S))-\Tr(H_S\rho_S)\leq 0.
\end{equation}
Recall that $\Tr(H_S\E_{H}(\rho_S))-\Tr(H_S\rho_S)$ gives an exchanged amount of heat from the heat bath at temperature $\beta_H$ (see Eq. \eqref{Defn:HeatExchanged}), which will be negative if the $\beta$-order of the initial state $\rho_S$ is negative. Employing the lemma \ref{Lemworkheat} further implies that the work done $W$ must be negative, which implies the engine cannot function if the initial state of the working body $\rho_S$ has a $\beta$-order $(21)$.

Therefore, in order to extract work, the initial state of the two-level working body should have $\beta$-order $(12)$ wrt. the hot heat bath, which allows us to write from Eq. \eqref{C2} the following:
\begin{equation}\label{betaorder12}
    p\geq \frac{1-p}{e^{-\beta_H\omega}} \;\;\Rightarrow\;\;e^{-\beta_H\omega} \geq \frac{1-p}{p}.
\end{equation}
Then, we can re-parameterize the state of the two-level working body $\v p$ using the effective temperature $\beta_C^{v}$ defined by the following relation \cite{BrunnerVirtualqubit2012}:
\begin{eqnarray}\label{virtualtemp}
\frac{1-p}{p} = e^{-\beta_C^{v}\omega}.
\end{eqnarray}
Using this parametrization, we can compute $p = \frac{1}{1+e^{-\beta_C^{v}\omega}}$. Furthermore, using Eqs. \eqref{betaorder12}  and \eqref{virtualtemp}, we can easily see $\beta_C^{v}\geq\beta_H$. 

Here, we would like to make a crucial remark. We can imagine any three-stroke engine with a two-level working body and population vector $\v p$ associated with two thermal baths at inverse temperatures $\beta_H$ and $\beta_C$ as an open-cycle engine that is associated with two thermal baths at inverse temperatures $\beta_H$ and $\beta^v_C$ (Recall that an open cycle engine is a three-stroke engine where the final stroke that aims to close the thermodynamic cycle is given by thermalization, as described in Sec. \ref{Open_cycle_engine}. We illustrate this idea in Fig. \ref{fig:Open_vs_closed}.) In other words, any stroke-based discrete heat engine with a two-level working body is essentially an open-cycle heat engine wrt. a new effective inverse temperature of the cold heat bath. 

Now, theorem 6 of Ref. \cite{Biswas2022extractionof} states that the optimal work output and efficiency for an open cycle engine are obtained at the extremal points of the thermal cone $\mathcal{T}_{H}(\v p)$, where the  unitary operation  in the work stroke is provided by 
\begin{equation}\label{NOT_for_work}
    X = \begin{pmatrix}
        0 & 1\\
        1 & 0
    \end{pmatrix},
\end{equation}
where $\T_{H}(\v p)$ is the thermal cone at the inverse temperature $\beta_{H}$ which justifies the subscript (Recall that thermal cone is defined as set of all states that can be achieved from $\rho_S$ via thermal operation at inverse temperature $\beta$, as described in Sec. \ref{Mathematical_tools}). As a result of the correspondence between the three-stroke engine and the open-cycle engine, we do not need to optimize the three-stroke engine's performance over the heat and work strokes.

Now, given a fixed $\lambda^{\text{max}}_H\in [0,1]$, the extremal points of $\mathcal{T}_{H}(\v p)$ are given by 
\begin{equation}\label{extreme_pts_of_TC}
    \Big\{\v p, \lambda_{H}^{\text{max}}\begin{pmatrix} 1-pe^{-\beta_H\omega} \\ pe^{-\beta_H\omega} \end{pmatrix}+(1-\lambda_{H}^{\text{max}})\v p\Big\},
\end{equation}
where the second extremal point can only lead to non-trivial work production and efficiency. Observe that, when $\lambda_{H}^{\text{max}} = 1$, we obtain the full thermal cone as given in Ref. \cite{horodecki2013fundamental}.

Therefore, optimal work production and efficiency have the following forms
\begin{eqnarray}
    W^{*}(p) &=&  \omega  \left(1-2\lambda_{H}^{\text{max}} +2 p (\lambda_{H}^{\text{max}} e^{-\beta_H \omega }+\lambda_{H}^{\text{max}}-1)\right),\label{work_star}\\
\eta^{*}(p) &=& 1-\frac{2p-1-\lambda_H^{\text{max}}\left(pe^{-\beta_H\omega}+p-1\right)}{\lambda_{H}^{\text{max}}(pe^{-\beta_H\omega}+p-1)}\label{eff_star}.
\end{eqnarray}
From Eq. \eqref{work_star}, it is straightforward to see that $W^*$ is an increasing function of $p$. On the other hand, by differentiating $\eta^*$ given in Eq. \eqref{eff_star} wrt. $p$, one can also see that $\eta^{*}(p)$ is an increasing function of $p$:
\begin{equation}
    \frac{\partial\eta^{*}(p)}{\partial p}= \frac{1-e^{-\beta_H\omega}}{\lambda^{\text{max}}_H\big(pe^{-\beta_H\omega}+p-1\big)^2}>0.
\end{equation}
Therefore, here our interest is to find out the highest possible value of $p$ that is compatible with the cyclicity condition. In order to do so, we proceed as follows. 

As the optimal efficiency and work produced by the engine achieved at the extremal point of the set $\T_{H}(\v p)$, thus
from Eq. \eqref{extreme_pts_of_TC}, one can see that the population vector of the $\v{p}$ associated with the initial state transformed to the population vector 
\begin{equation}\label{resulting_1}
    \lambda_{H}^{\text{max}}\begin{pmatrix} 1-pe^{-\beta_H\omega} \\ pe^{-\beta_H\omega} \end{pmatrix}+(1-\lambda_{H}^{\text{max}})\v p = \lambda_{H}^{\text{max}}\begin{pmatrix} 1-pe^{-\beta_H\omega} \\ pe^{-\beta_H\omega} \end{pmatrix}+(1-\lambda_{H}^{\text{max}})\begin{pmatrix} p \\ 1-p \end{pmatrix},
\end{equation}
during the heat stroke. After the work stroke which is governed by the unitary operation $X$ in Eq. \eqref{NOT_for_work}, the resulting population vector from the heat stroke given in Eq. \eqref{resulting_1} transformed to 
\begin{equation}\label{defz}
    \lambda_{H}^{\text{max}}\begin{pmatrix} pe^{-\beta_H\omega} \\ 1-pe^{-\beta_H\omega} \end{pmatrix}+(1-\lambda_{H}^{\text{max}})\begin{pmatrix} 1-p \\ p \end{pmatrix} = \begin{pmatrix} \lambda_{H}^{\text{max}}pe^{-\beta_H\omega}+(1-\lambda^{\text{max}}_H)p \\ \lambda_{H}^{\text{max}}\left(1-pe^{-\beta_H\omega}\right)+(1-\lambda^{\text{max}}_H)(1-p) \end{pmatrix}:=\begin{pmatrix} z \\ 1-z \end{pmatrix} = \v z,
\end{equation}
where $z=\lambda_{H}^{\text{max}}pe^{-\beta_H\omega}+(1-\lambda^{\text{max}}_H)p$.
Now the final stroke which is a thermal operation wrt. a cold bath, should transform the population vector $\v z$ to the initial state which results in closing the cycle. Therefore,
\begin{equation}\label{defz2}
    \lambda_C\begin{pmatrix} 1-ze^{-\beta_C\omega} \\ ze^{-\beta_C\omega} \end{pmatrix}+(1-\lambda_C)\begin{pmatrix} z \\ 1-z \end{pmatrix} = \begin{pmatrix} p \\ 1-p \end{pmatrix}.
\end{equation}
By substituting $z$ from Eq. \eqref{defz} in Eq. \eqref{defz2} and solving the resulting equation for $p$, we obtain 
\begin{equation}\label{p_intermediate}
    p = \frac{1-\lambda^{\text{max}}_H(1-\lambda_C)-\lambda_C(1-\lambda^{\text{max}}_H)e^{-\beta_C\omega}}{2-\lambda_C(1-\lambda_H^{\text{max}})-\lambda_H^{\text{max}}-\lambda_C(1-\lambda_H^{\text{max}})e^{-\beta_C\omega}-\lambda_H^{\text{max}}(1-\lambda_C)e^{-\beta_H\omega}+\lambda_H^{\text{max}}\lambda_Ce^{-(\beta_C+\beta_H)\omega}}.
\end{equation}
We can see that $p$ is an increasing function of $\lambda_C$ by differentiating $p$ wrt. $\lambda_C$ as
\begin{equation}\label{above}
    \frac{\partial p}{\partial \lambda_C} =  \frac{e^{\omega  (\beta_C+\beta_H)} \Big(e^{\omega  (\beta_C+\beta_H)}-\big(\lambda^{\text{max}}_H e^{\beta_C \omega }+(1-\lambda^{\text{max}}_H) e^{\beta_H\omega}\big)\Big)}{\bigg(\Big(2-\lambda^{\text{max}}_H-(1-\lambda_H^{\text{max}})\lambda_C\Big)e^{\omega(\beta_C+\beta_H)}-(1-\lambda_C)\lambda_H^{\text{max}}e^{\beta_C\omega}-(1-\lambda_H^{\text{max}})\lambda_Ce^{\beta_H\omega}+\lambda_C\lambda_H^{\text{max}}\bigg)^2} > 0.
\end{equation}
The inequality in the above Eq. \eqref{above} follows because $e^{\beta_C\omega}>1$ and $e^{\beta_H\omega}>1$, therefore, the second term in the numerator appeared on RHS of Eq. \eqref{above} is positive, i.e.,
\begin{equation}
    e^{\omega  (\beta_C+\beta_H)}-\big(\lambda^{\text{max}}_H e^{\beta_C \omega }+(1-\lambda^{\text{max}}_H) e^{\beta_H\omega}\big)>0,
\end{equation}
which implies that the numerator is always positive, whereas the denominator is always positive because of the square. Therefore, the largest possible value of $p$ compatible with the cyclicity condition is obtained when $\lambda_C =\lambda_C^{\text{max}}$, i.e.,
\begin{equation}\label{cyclicityp}
    p = \frac{1-\lambda^{\text{max}}_H(1-\lambda_C^{\text{max}})-\lambda_C^{\text{max}}(1-\lambda^{\text{max}}_H)e^{-\beta_C\omega}}{2-\lambda_C^{\text{max}}(1-\lambda_H^{\text{max}})-\lambda_H^{\text{max}}-\lambda_C^{\text{max}}(1-\lambda_H^{\text{max}})e^{-\beta_C\omega}-\lambda_H^{\text{max}}(1-\lambda_C^{\text{max}})e^{-\beta_H\omega}+\lambda_H^{\text{max}}\lambda_C^{\text{max}}e^{-(\beta_C+\beta_H)\omega}}.
\end{equation}
Substituting the value of $p$ from Eq. \eqref{cyclicityp} in Eqs. \eqref{work_star} and \eqref{eff_star}, we get 
\begin{eqnarray}
W_{\text{max}} &=& 1-2\lambda^{\text{max}}_H+\frac{2\Big(\lambda_He^{-\beta_H\omega}-(1-\lambda_H^{\text{max}})\Big)\Big(1-(1-\lambda_C^{\text{max}})\lambda_H^{\text{max}}-\lambda_C^{\text{max}}e^{-\beta_C\omega}(1-\lambda_{H}^{\text{max}})\Big)}{2-\lambda_C^{\text{max}}(1-\lambda_H^{\text{max}})-\lambda_H^{\text{max}}-\lambda_C^{\text{max}}e^{-\beta_C\omega}(1-\lambda_H^{\text{max}})-\lambda_H^{\text{max}}e^{-\beta_H\omega}(1-\lambda_C^{\text{max}})+\lambda_{C}^{\text{max}}\lambda_{H}^{\text{max}}e^{-(\beta_C+\beta_H)\omega}},\nonumber\\
\eta_{\text{max}} &=& 1-\frac{\lambda_C^{\text{max}}\Big(1-\lambda_H^{\text{max}}e^{-\beta_H\omega}-(1-\lambda_H^{\text{max}})e^{-\beta_C\omega}\Big)}{\lambda^{\text{max}}_H\Big(e^{-\beta_H\omega}-(1-\lambda^{\text{max}}_C)-\lambda_C^{\text{max}}e^{-(\beta_H+\beta_C)\omega)}\Big)}.
\end{eqnarray}
This completes the proof of theorem \ref{thm_optimal_work_eff}.

\section{Range of thermal operation with finite sized baths}\label{Appendix_finite_sized_baths}

In this section, we describe the range of thermal operations for a qubit system with a finite size heat bath. The thermal operations on a qubit system can be expressed as a convex mixture of two extremal thermal processes $\lambda T^{\E}+(1-\lambda)\iden$, where $0\leq \lambda \leq 1$ and
\begin{equation}
    \mathcal{E}=\begin{pmatrix}
    1-e^{-\beta\omega} & 1\\
    e^{-\beta\omega} & 0
    \end{pmatrix}.
\end{equation} 
However, when the bath is of finite dimension the full range of $\lambda$ is not allowed. Let's consider the dimension of the heat bath is $d_E=d+1$. Tha Hamiltonians of the system and heat bath are
\begin{eqnarray}
H_S=\omega\ket{1}\!\bra{1}\,\, \mbox{and}\,\,
H_E=\omega\sum_{n=0}^d n \ket{n}\!\bra{n}
\end{eqnarray}
respectively. Note that we assume that the ground state energies of the system and heat bath are zero, and the energy levels are eqispaced. With this we can state the following theorem:
\begin{thm}\label{thm_range_thermal}
The allowed range of thermal operation for a qubit system when the heat bath is finite, is
$\lambda\in\{0,\lambda^{\max}\}$, where
\begin{equation}
    \lambda^{\max}=\frac{1-e^{-\beta\omega d}}{1-e^{-\beta\omega(d+1)}}.
\end{equation} 
\end{thm}
\begin{proof}
 As $[U, H_S\otimes \iden_E+ \iden_S\otimes H_E] = 0$, then for a thermal operation with such a heat bath, $U$ has the following $2\times 2$ block-diagonal structure in the energy eigenbasis:
\begin{equation}\label{unitary_representation}
   U=\iden_{\{\ket{0}_S\ket{0}_{E}\}}\bigoplus_{i=1}^d V_{i_{\{\ket{0}_S\ket{i}_{E},\ket{1}_S\ket{i-1}_{E}\}}}\oplus\iden_{\{\ket{1}_S\ket{d}_{E}\}},
\end{equation}
where $V_i$'s are $2\times2$ unitary matrix in the sub-space $\{\ket{0}_S\ket{i}_{E},\ket{1}_S\ket{i-1}_{E}\}$ and $\iden$'s are basically 1. Further, we take initial system state $\rho_S$ to be diagonal in the energy eigenbasis, i.e.,
\begin{equation}\label{initial_system}
    \rho_S=
    \begin{pmatrix}
    p & 0 \\
    0 & 1-p
    \end{pmatrix}.
\end{equation}
The state of the bath is
\begin{equation}
    \rho_{E}=\frac{e^{-\beta H_{E}}}{Z_{H_{E}}}=\frac{1}{Z_{H_E}}
    \begin{pmatrix}
    1 & 0 & \cdots & 0 \\
    0 & e^{-\beta\omega} & \cdots & 0\\
    \vdots & \vdots & \vdots & \vdots \\
    0 & 0 & \cdots & e^{-d\beta\omega}
    \end{pmatrix},
\end{equation}
where $Z_{H_E}=\mbox{Tr}\left(e^{-\beta H_E}\right)$. Hence, the initial state of the system and bath has the following form in the energy eigenbasis (see equation Eq. (\ref{unitary_representation})):
\begin{equation}
    \rho_S\otimes\rho_E=\frac{1}{Z_{H_E}}=
    p \bigoplus_{i=1}^d 
    \begin{pmatrix}
    p e^{-i\beta\omega} & 0\\
    0 & (1-p)e^{-(i-1)\beta\omega}
    \end{pmatrix}
    \oplus (1-p)e^{-d \beta \omega}.
\end{equation}
Note that the energies of the blocks from left to right are $0, \omega, 2 \omega, \ldots, (d+1)\omega$. As the unitary operation $\left(U\right)$ and the initial state of the system and bath have similar block diagonal structures, $U\left(\rho_S\otimes\rho_E\right)U^{\dagger}$ also has a similar block diagonal structure. Therefore, it is enough to find out the elements of one block. Consider a block position $j$ with energy $j\omega$, then we have
\begin{eqnarray}
   \rho_{SE}&=& V_j (\rho_S\otimes\rho_E) V_j^\dagger=
    \begin{pmatrix}
    e^{i\phi_j}\cos\theta_j & e^{i\alpha_j}\sin\theta_j\\
    -e^{-i\alpha_j}\sin\theta_j &
    e^{-i\phi_j}\cos\theta_j
    \end{pmatrix}
    \begin{pmatrix}
    p e^{-j\beta\omega} & 0 \\
    0 & (1-p)e^{-(j-1)\beta\omega}
    \end{pmatrix}
    \begin{pmatrix}
    e^{i\phi_j}\cos\theta_j & -e^{i\alpha_j}\sin\theta_j\\
    e^{-i\alpha_j}\sin\theta_j &
    e^{i\phi_j}\cos\theta_j
    \end{pmatrix}\\
    &=&
    \begin{pmatrix}
    p e^{-j\beta\omega}\cos^2\theta_j+(1-p)e^{-(j-1)\beta\omega}\sin^2\theta_j & e^{i(\phi_j+\alpha_j)}\left(
    e^{-(j-1)\beta\omega}(1-p)-pe^{-j\beta\omega}\right)\sin\theta_j\cos\theta_j \\
    e^{-i(\phi_j+\alpha_j)}\left(
    e^{-(j-1)\beta\omega}(1-p)-pe^{-j\beta\omega}\right)\sin\theta_j\cos\theta_j &  p e^{-j\beta\omega}\sin^2\theta_j+(1-p)e^{-(j-1)\beta\omega}\cos^2\theta_j
    \end{pmatrix},\nonumber
\end{eqnarray}
where $V_j$ is an arbitrary $2\times2$ special unitary matrix. After tracing out the thermal bath from the final state $\rho_{SE}$, we get the final state of the system as follows
\begin{equation}
    \rho'_S=\frac{1}{Z_{H_E}}
    \begin{pmatrix}
    \mathcal{A} & 0\\
    0 & Z_{H_E}-\mathcal{A}
    \end{pmatrix},
\end{equation}
where $\mathcal{A}=p+\sum_{j=1}^d \left(p e^{-j\beta\omega}\cos^2\theta_j+(1-p)e^{-(j-1)\beta\omega}\sin^2\theta_j\right)$.

On the other hand, under thermal operations as described above theorem \ref{thm_range_thermal}, one of the diagonal elements of the state $\rho_S$ changes to
\begin{eqnarray}
p \rightarrow p'= 1-\left((1-p)(1-\lambda)+e^{-\beta\omega}p\lambda\right). 
\end{eqnarray}
To find the solution of $\lambda$, we need to compare $p'$ with $\mathcal{A}/Z_{H_E}$, i.e.,
\begin{equation}
    \lambda = \frac{1}{Z_H}\sum_{j=1}^d \sin^2\theta_j e^{-\beta(j-1)\omega}. 
\end{equation}
The maximum value of $\lambda$ is obtained when $\sin^2\theta_j=1$ $\forall \theta_j$ and the value is
\begin{equation}
  \lambda^{\max}= \sum_{j=1}^d\frac{e^{-\beta(j-1)\omega}}{Z_{H}} = \frac{1-e^{-\beta\omega d}}{1-e^{-\beta\omega(d+1)}}.
\end{equation}
This completes the proof.
\end{proof}
Therefore, the range of $\lambda$ is $\{0,\lambda^{\max}\}$. It is easy to check that for $d\rightarrow\infty$, $\lambda^{\max}=1$, and thus all thermal operations are possible. If the heat bath is a qubit, then $\lambda^{\max}=1/\left(1+e^{\beta\omega}\right)$, which implies that the thermal operation corresponds to $\lambda^{\max}$ just thermalizes any input system. 

\twocolumngrid
\bibliography{main}

\end{document}